\documentclass[a4paper]{article}

\usepackage[ansinew]{inputenc}
\usepackage{amsmath}
\usepackage{amsthm}
\usepackage{amssymb}
\usepackage{amsfonts}
\usepackage{color}
\usepackage{tikz}

\newcounter{saetning} 
  \newtheorem{theo}[saetning]{Theorem}
  \newtheorem{coro}[saetning]{Corollary}
  \newtheorem{lemm}[saetning]{Lemma}
  \newtheorem{prop}[saetning]{Proposition}
\theoremstyle{definition} 
	\newtheorem{defi}{Definition}
\theoremstyle{remark}

\definecolor{gray}{gray}{0.6}

\newcommand{\N}{\mathbb{N}}

\newcommand{\E}{\mathbb{E}}

\title{A numbers-on-foreheads game}
\begin{document}
\date{}
\author{Sune K. Jakobsen\thanks{School of Mathematical Sciences and School of Electronic Engineering \&
Computer Science, Queen Mary University of London, Mile End Road, London,
E1 4NS, UK. Email: S.K.Jakobsen@qmul.ac.uk.}}
  \maketitle

  \section*{Abstract}
  Is there a joint distribution of $n$ random variables over the natural numbers, such that they always form an increasing sequence and whenever you take two subsets of the set of random variables of the same cardinality, their distribution is almost the same?
  
    We show that the answer is yes, but that the random variables will have to take values as large as $2^{2^{\dots ^{2^{\Theta\left(\frac{1}{\epsilon}\right)}}}}$, where $\epsilon\leq \epsilon_n$ measures how different the two distributions can be, the tower contains $n-2$ $2$'s and the constants in the $\Theta$ notation are allowed to depend on $n$. This result has an important consequence in game theory: It shows that even though you can define extensive form games that cannot be implemented on players who can tell the time, you can have implementations that approximate the game arbitrarily well.

  \section{Introduction}

  A group of gamblers are standing in a circle so that each gambler can see all the other gamblers' foreheads but not their own, and the gamblers are not allowed to communicate. A dealer then sticks one natural number on each gamblers' forehead, and ask each gambler to choose two numbers $i$ and $j$ (the other gambler do not learn the numbers). If the gambler had the $i$'th smallest number he wins $1$ dollar from the dealer, if he has the $j$'th smallest he losses $1$ dollar to the dealer. Does she have a randomised strategy that ensures that in expectation she does not lose money? 
  
  A randomised dealer strategy is just a distribution on $(X_1,\dots,X_n)$ where $n$ is the number of gamblers and $X_1<X_2<\dots<X_n$ are random variables taking integers. We show that if such a strategy ensures that the dealer does not lose in expectation, then for any $k$ any two $k$-subsets of $\{X_1,\dots ,X_n\}$ will have the same distribution. This has an important consequence in game theory. It is well known that some extensive form games, for example the absent minded driver game~\cite{PR97}, cannot be implement on agents with perfect memory. In order to avoid such games, we often require games to have \emph{perfect recall}, that is, the players remember which information sets they have previously been in, and what choices they made. However, there are also games with perfect recall that cannot be implemented on agents with a sense of time, because the information set do not respect any ordering~\cite{Time}. If we can find a distribution on $(X_1,\dots, X_n)$ such that learning the values of some subset $\{X_{i_1},\dots,X_{i_k}\}$ only tells you the \emph{cardinality} $k$ of the subset, but does not give you any information about $i_1,\dots, i_k$, then we can use this to timing any game with perfect recall and at most $n$ nodes in each history. We simply play the root at time $X_1$, then the next node at time $X_2$ and so on. The agents would learn some times $X_{i_j}$, but the only information the agents would get from this, is the number of nodes he has had, and agents alway know this in any game with prefect recall. 
  
  However, we show that no such distribution of $(X_1,\dots,X_n)$ exists. In the other direction, we show that the dealer can ensure that she only lose $\epsilon$ dollars in expectation for any $\epsilon>0$. Unfortunately, to do that, she will need to numbers as large as $2^{2^{\dots ^{2^{\Theta(\frac{1}{\epsilon})}}}}$, when $\epsilon$ is sufficiently small. Here the tower contains $n-2$ $2$'s, and the constants in the $\Theta$ notation are allowed to depend on $n$. This result implies that any extensive form game can be approximated arbitrarily well by games where the players know the time at any node. The rest of this paper is about the numbers-on-foreheads game and related problems. For the game theoretical implications of the paper, see~\cite{Time}.
   
   To analyse the problem, we define a slightly different game. First we let the Dealer choose some distribution of $(X_1,\dots, X_n)$ such that $X_1,\dots ,X_n$ is a strictly increasing sequence of natural numbers. Then a gambler chooses two numbers $i$ and $j$. 
   Then $X=(X_1,\dots,X_n)$ is chosen randomly using the distribution given by the dealer, and an independent coin flip decides if the gambler is given $X_{-i}:=(X_1,\dots ,X_{i-1},X_{i+1},\dots, X_n)$ or $X_{-j}$. 
   The gambler then bets $1$ dollar on which of the two he was given. The expected utility of the gambler in this within a factor $n$ of the expected utility for the gambler of the first game. 
   Furthermore, this game is easier to analyse: Gambler cannot win more than $\epsilon$ is expectation if and only if any two $X_{-i}$ and $X_{-j}$ have total variation distance at most $\epsilon$.

   \subsection{Notation}
     For real numbers $x,y$ we define $[x]=\{i\in \N| i\leq x\}$, and $[x,\dots ,y]=\{i\in\N| x\leq i\leq y\}$. We let $\log$ denote the base $2$ logarithm, and let $\exp_2$ denote the function $x\mapsto 2^x$. Hence, $\exp_2^n(x)$ denotes iteration of $\exp_2$, so $\exp_2^n(x)=2^{2^{\dots^{2^x}}}$ where the tower contains $n$ $2$'s. Similarly, $\log^n$ denotes iteration of $\log$.

  %For increasing functions $g$ and decreasing functions $h$ we use $f(\epsilon)=g(\Omega(h(\epsilon)))$ to mean 
 % \[\exists c,\epsilon_0\forall \epsilon\in(0,\epsilon_0): f(\epsilon)\geq g(ch(\epsilon)).\]
 % This is different from saying that there exists a function $h'=\Omega(h)$ such that $f(\epsilon)=g(h'(\epsilon))$: even if $g$ do not take values below e.g $2^{2^2}$ and $f$ does take smaller values we can still have $f(\epsilon)=g(\Omega(h(\epsilon)))$ according to our notation.
  
%  Similarly, $f(\epsilon)=g(O(h(\epsilon)))$ means
%  \[\exists C,\epsilon_0\forall \epsilon\in (0,\epsilon_0): f(\epsilon)\leq g(Ch(\epsilon))\]
%  and $f(\epsilon)=g(\Theta(h(\epsilon)))$ means that both $f(\epsilon)=g(O(h(\epsilon)))$ and $f(\epsilon)=g(\Omega(h(\epsilon)))$.

  The \emph{total variation distance} (also called \emph{statistical distance}) between two discrete random variables $X_1$ and $X_2$ is given by
  \begin{align*}\delta(X_1,X_2)&=\sum_{x}\max(\Pr(X_1=x)-\Pr(X_2=x),0)\\
  &=\sum_x\left|\frac{\Pr(X_1=x)-\Pr(X_2=x)}{2}\right|
  \end{align*}
  where the sums are over all possible values of $X_1$ and $X_2$. This measure is symmetric in $X_1$ and $X_2$.
  
  We say that $X_1$ and $X_2$ are \emph{$\epsilon$-indistinguishable} if $\delta(X_1,X_2)\leq \epsilon$. 
  Given a tuple $X=(X_1,\dots, X_n)$ we let $X_{-i}$ denote $(X_1,\dots, X_{i-1},X_{i+1},\dots,X_n)$. We say that $(X_1,\dots, X_n)$ \emph{has $\epsilon$-indistinguishable $m$-subsets} if for any two subsets $\{i_1,\dots,i_m\},\{j_1,\dots j_m\}\subset [n]$ of size $m$, the two random sets $\{X_{i_1},\dots ,X_{i_n}\}$ and $\{X_{j_1},\dots, X_{j_n}\}$ are $\epsilon$-indistingushable. We slightly abuse notation and say that $(X_1,\dots, X_n)$ \emph{has $\epsilon$-indistinguishable subsets} if for all $m<n$ it has $\epsilon$-indistinguishable $m$-subsets.
     
We will assume that the reader knows some game theory, and knows the minimax theorem. For an introduction to game theoretical concepts see~\cite{AGT07}. 
    
 \subsection{Paper outline}

In the next section we show that the total variation distance can be used to measure the advantage you get from side information when entering an otherwise fair bet, and in Section \ref{sec:props} we show some properties of total variation distance that we will need later. Then we ask, can we find a distribution of $(X_1,\dots, X_n)$ with $\epsilon$-indistinguishable subsets, and if so, how large values does $X_n$ need to take. In Section \ref{sec:construction} we give a recursive construction of such a tuple for any $n$ and $\epsilon>0$. By the results in Section \ref{sec:gambling} this corresponds to good randomised strategy for Dealer. In Section \ref{sec:lower bound} we show a lower bound on how large numbers $X_n$ needs to take, and hence how large numbers Dealer needs to use, to ensure that she only losses $\epsilon$ in expectation. 

\section{Relation between gambling games and total variation distance}\label{sec:gambling}

In this section we will show that several similar problems are the same up to a constant factor (depending on $n$) on $\epsilon$. First we show that the total variation distance can be seen as an measure of the advantage in a betting game.

\begin{prop}[Total variation as betting advantage]\label{prop:tvdasbetadv}
For random variables $Y_1$ and $Y_2$ we define a one-player game:
\begin{itemize}
\item $y_1$ and $y_2$ is chosen according to the distribution of $Y_1$ and $Y_2$.
\item Independently $i$ is chosen uniformly on $\{1,2\}$.
\item The player learns $y_i$ and makes a guess about $i$.
\item If correct he gets utility $1$ if wrong he gets utility $-1$.  
\end{itemize}
The expected utility the player gets using the optimal strategy is $\delta(Y_1,Y_2)$
\end{prop}
\begin{proof}
As this is a one-player game, the optimal strategy is deterministic. A deterministic strategy is a function $g$ that for each possible value $y$ of $Y_1$ or $Y_2$ gives the value in $\{1,2\}$ that the player should guess. If $g(y)=1$ the contribution of $y$ to the expected output when gambler use strategy $g$ is
\[\Pr(Y_1=y)\Pr(I=1)-\Pr(Y_2=y)\Pr(I=2)=\frac{\Pr(Y_1=y)-\Pr(Y_2=y)}{2}.\]
Similarly, if $g(y)=2$ then $y$'s contribution to the expected outcome is
\[\frac{\Pr(Y_2=y)-\Pr(Y_1=y)}{2}.\]
Clearly, the best strategy is to choose the positive one of these two, in which case the contribution of $y$ is
\[\left|\frac{\Pr(Y_1=y)-\Pr(Y_2=y)}{2}\right|.\]
Summing over all $y$'s gives $\delta(Y_1,Y_2)$.
\end{proof}

We will now define two games between Dealer and Gambler, and show that they are related.
Given $n$ and $N$ we define Game 1: 
\begin{itemize}
\item Dealer chooses some natural numbers $1\leq x_1<\dots<x_n\leq N$. 
\item A number $i_0$ is chosen uniformly at random from $[n]$.
\item Gambler learns $x_{-i_0}$.
\item Gambler chooses two numbers $i$ and $j$.
\item If $i=i_0$ the Gambler wins $1$ dollar from Dealer, if $j=i_0$ he loses $1$ dollar to Dealer.
\end{itemize}
This is just the game from the introduction seen from the perspective of a single gambler. 

Game 1 is a two-player zero-sum games, and because we only allow the dealer to choose natural numbers between $1$ and $N$, each of these players only have finitely many pure strategies. By the minimax theorem such a game has a value $v_1$ such that
\begin{itemize}
\item Dealer has a probabilistic strategy, that is a distribution on $(X_1,\dots, X_n)$, such that no matter what strategy Gambler uses he cannot earn more than $v_1$ dollars in expectation.
\item Gambler has a probabilistic strategy such that no matter which numbers Dealer chooses, Gambler will win at least $v_1$ dollars in expectation.
\end{itemize}

We want to figure out how this value changes with $n$ and $N$. In order to do this we define Game 2, which is less natural but easier to analyse:
\begin{itemize}
\item Dealer chooses some natural numbers $1\leq x_1<\dots<x_n\leq N$. 
\item Gambler chooses two numbers $i_1$ and $i_2$.
\item A fair coin is flipped to decide if $K=1$ or $2$.
\item Gambler learns $X_{-i_K}$ and guesses if $K=1$ or $2$.
\item If Gambler is correct he wins $1$ dollar from Dealer, if he is wrong he loses $1$ dollar from Dealer.
\end{itemize}
This is also a zero-sum game with finitely many pure strategies, so the minimax theorem says that Game 2 also have a value $v_2$. We want to show that $v_1$ and $v_2$ are within a factor $n$ of each other.

\begin{prop}\label{prop:v1boundsv2}
$\frac{2}{n}v_2\leq v_1$
\end{prop}
\begin{proof}
Suppose that Gambler has a strategy that ensures an expected outcome of $v_2$ in Game $2$. To show the statement we construct a strategy that ensures expected outcome of $\frac{2}{n}v_2$ in Game $1$. 

Before Game $1$ starts, the Gambler chooses two numbers $i_1$ and $i_2$ using his strategy for Game $2$. We then play Game $1$: Dealer chooses numbers $x_1<\dots<x_n$ and $i_0$ is chosen uniformly from $[n]$. Gambler sees $x_{i_0}$. The Gambler still plays as if he was playing Game $2$ and had chosen $i_1$ and $i_2$. If he would have guessed $K=1$ he sets $i=i_1, j=i_2$ and if he would have guessed $2$ he sets $i=i_2,j=i_1$. 

As Gambler choice of $i_1$ and $i_2$ cannot affect $i_0$ and $X$, there is probability $\frac{n}{2}$ that $i_0\in\{i_1,i_2\}$. Given that this happens the expected outcome is exactly the same as the expected outcome in Game $2$. Given that $i_0\not\in\{i_1,i_2\}$ the Gambler will neither lose nor win money. Thus, the expected outcome of the strategy is $\frac{2}{n}v_2$. 
\end{proof}

\begin{prop}\label{prop:v2boundsv1}
$\frac{v_1}{n-1}\leq v_2$
\end{prop}
\begin{proof}
Suppose that Gambler has a strategy that ensures an expected outcome of $v_1$ in Game $1$. To show the statement we construct a strategy that ensure expected outcome of $\frac{v_1}{n-1}$ in Game $2$. To do this we define a Game $1.5$ where the Gambler starts by choosing $i_1$ and $i_2$ as in Game $2$, but he is given $X_{-i_0}$ where $i_0$ is uniformly distributed on $[n]$ independently from $X$. If $i_0\neq i_1,i_2$ no one wins anything. More formally:
\begin{itemize}
\item Dealer chooses some natural numbers $1\leq x_1<\dots<x_n\leq N$. 
\item Gambler chooses two numbers $i_1$ and $i_2$.
\item $i_0$ is chosen uniformly at random from $[n]$.
\item Gambler learns $X_{-i_0}$ and guesses if $i_0=i_1$ or $i_2$.
\item If $i_0\not\in\{i_1,i_2\}$ no money is transferred, otherwise
\item If Gambler guessed correct he wins $1$ dollar from Dealer, if he guessed wrong he loses $1$ dollar to Dealer.
\end{itemize} 

We convert the gambler strategy for Game $1$ to a Gambler strategy for Game $1.5$: Gambler first chooses $\{i_1,i_2\}$ uniformly from all size-$2$ subsets of $[n]$. Then when he sees $X_{-i_0}$ he considers which two number he would have chosen as $i$ and $j$ in Game $1$. If $\{i,j\}=\{i_1,i_2\}$ he makes his bet as in Game $1$. This case happens with probability $\frac{2}{n(n-1)}$, and given that it happens he has excepted outcome $v_1$. Otherwise, Gambler chooses the bet $i_0=i_1$ with probability $0.5$ and $i_0=i_2$ with probability $0.5$, to ensure that he has excepted outcome $0$ in this case. Thus, expected outcome using that strategy for Game $1.5$ is $\frac{2v_1}{n(n-1)}$.

We now use the same strategy for Game $2$. Only difference is that we know that $i_0\in\{i_1,i_2\}$, thus we get the expected outcome from game $1.5$ conditioned on $i_0\in\{i_1,i_2\}$. But if $i_0\not\in\{i_1,i_2\}$ the outcome of game $1.5$ is $0$, thus the entire contribution to the expected outcome of game $1.5$ comes from the case $i_0\in\{i_1,i_2\}$. This case happens with probability $\frac{2}{n}$, so the expected value of our strategy in Game $2$ is $\frac{2v_1}{n(n-1)}\frac{n}{2}=\frac{v_1}{n-1}$.
\end{proof}

\begin{prop}\label{prop:v_2=eps}
$v_2\leq \epsilon$ if and only if there exists $(X_1,\dots, X_n)$ with $1\leq X_1<\dots X_n\leq N$ with $\epsilon$-indistinguishable $n-1$-subsets.
\end{prop}
\begin{proof}
$\Rightarrow$: Assume that $v_2\leq\epsilon$. Then there exists a mixed strategy for Dealer that ensures that Gambler wins at most $\epsilon$. This mixed strategy is a distribution of $(X_1,\dots,X_n)$ and by assumption, no matter which $i,j$ Gambler choses, he cannot win more than $\epsilon$ in expectation. By Proposition \ref{prop:tvdasbetadv} this mean that $\delta(X_{-i},X_{-j})\leq \epsilon$ for all $i,j$.

$\Leftarrow$: Assume that there exists $(X_1,\dots,X_n)$ as in the statement. Then the Dealer can use this distribution to choose her number in Game $2$ and by Proposition \ref{prop:tvdasbetadv} the Gambler cannot win more than $\epsilon$ in expectation with any strategy.
\end{proof}

\begin{theo}\label{theo:equivalence}
Fix parameters $n\geq 2$ and $N$.
Let $v_1$ be the value of Game $1$ and let $\epsilon_0$ be the supremum over all values $\epsilon$ such that there exists a distribution of $X=(X_1,\dots,X_n)$ where $1\leq X_1<X_2<\dots<X_n\leq N$ are all integers and $X$ has $\epsilon$-indistinguishable $n-1$-subsets. Then $\frac{2}{n}\epsilon_0 \leq v_1\leq (n-1)\epsilon_0 $
\end{theo}
\begin{proof}
By Proposition \ref{prop:v_2=eps} we have $v_2=\epsilon_0$. The theorem now follows from Proposition \ref{prop:v1boundsv2} and \ref{prop:v2boundsv1}. 
\end{proof}

All the above problems have only involved $n-1$-subsets on the $n$ numbers. The following proposition show that if all $n-1$-subsets look the same, then any two subsets of the same size looks the same.

\begin{prop}\label{prop:subsetfromnminus1}
Fix $n\in \N$ and $\epsilon>0$. If $(X_1,\dots , X_n)$ has $\epsilon$-indistinguishable neighbouring $n-1$-subsets, is has $n^2\epsilon$-indistinguishable subsets. 
\end{prop}
\begin{proof}
Assume that $(X_1,\dots, X_n)$ has $\epsilon$-indistinguishable $n-1$-subsets. We say two subsets of $[n]$ of the same size are neighbours if only one of the numbers in them are different, and the different numbers only differ by one. That is, we can write them as $\{i_1,\dots, i_{k-1},i_k,i_{k+1},\dots i_m\}$ and $\{i_1,\dots i_{k-1},i_k+1,i_{k+1},\dots i_m\}$. We will assume that $i_1<i_2<\dots<i_m$. Let $f$ be the function given by
 $f(x_1,\dots, x_{n-1})=(x_{i_1},\dots x_{i_k-1},x_{i_k}, x_{i_{k+1}+1},\dots x_{i_m+1})$. 
 Now $f(X_1,\dots, X_{i_k-1},X_{i_k},X_{i_k+2},\dots , X_n)=(X_{i_1},\dots ,X_{i_{k-1}},X_{i_k},X_{i_{k+1}},\dots X_{i_m})$ 
 and $f(X_1,\dots, X_{i_k-1},X_{i_k+1},X_{i_k+2},\dots X_n)=(X_{i_1},\dots , X_{i_{k-1}},X_{i_k+1},X_{i_{k+1}},\dots, X_{i_m})$. 
The two string we use as argument in $f$ are neighbouring $n-1$ subsets of $[n]$, so by assumption they have statistical distance at most $\epsilon$. Thus by Proposition \ref{prop:dpi} the result must also have statistical distance at most $\epsilon$, so any two neighbouring $m$-sets have statistical distance at most $\epsilon$. In the graph where the nodes are $m$-sets and two nodes are connected if the $m$-sets are neighbours, the diameter is no more than $n^2$. By the triangle inequality (Proposition \ref{prop:triangle}) the statistical distance between any two $m$-sets is less than $n^2\epsilon$.
\end{proof}

\section{Properties of total variation distance}\label{sec:props}

 In this section we show some basic properties about the total variation distance. First the triangle inequality.

\begin{prop}[Triangle inequality]\label{prop:triangle}
For random variables $X_1,X_2,X_3$ we have
\[\delta(X_1,X_3)\leq \delta(X_1,X_2)+\delta(X_2,X_3)\]
\end{prop}
\begin{proof}
\begin{align*}
\delta(X_1,X_3)=&\sum_x \max(\Pr(X_1=x)-\Pr(X_3=x),0)\\
\leq & \sum_x\max(\Pr(X_1=x)-\Pr(X_2=x),0)+\max(\Pr(X_2=x)-\Pr(X_3=x),0)\\
=&\delta(X_1,X_2)+\delta(X_2,X_3). 
\end{align*}
\end{proof}

\begin{prop}\label{prop:flag}
Two random variables $X_1$ and $X_2$ have total variation distance at most $\epsilon$ if and only if there exists a joint distribution $(X_1,S)$ where $S$ takes values in $\{0,1\}$, $\Pr(S=0)\leq \epsilon$ and for all $x$, $\Pr((X_1,S)=(x,1))\leq \Pr(X_2=x)$.
\end{prop}

\begin{proof}
First assume that $X_1$ and $X_2$ have total variation distance at most $\epsilon$. Then we define a distribution of $(X_1,S)$ by $\Pr(S=1|X_1=x)=\min\left(1,\frac{\Pr(X_2=x)}{\Pr(X_1=x)}\right)$. Now we have $\Pr(S=1|X_1=x)\leq \frac{\Pr(X_2=x)}{\Pr(X_1=x)}$ so $\Pr(S=1,X_1=x)\leq \Pr(X_2=x)$. We have
\begin{align*}
\Pr(S=0)=&\sum_x \Pr(S=0|X_1=x)\Pr(X_1=x)\\
=&\sum_x \max\left(0,1-\frac{\Pr(X_2=x)}{\Pr(X_1=x)}\right)\Pr(X_1=x)\\
=&\sum_x\max\left(0,\Pr(X_1=x)-\Pr(X_2=x)\right)\\
=&\delta(X_1,X_2).
\end{align*} 

To show the opposite implication, assume that there is a distribution of $(X_1,S)$ as in the statement. Then 
\begin{align*}
\delta(X_1,X_2)=&\sum_x \max(\Pr(X_1=x)-\Pr(X_2=x),0)\\
=& \sum_x \max(\Pr(X_1=x,S=0)+\Pr(X_1=x,S=1)-\Pr(X_2=x),0)\\
\leq &\sum_x\max(\Pr(X_1=x,S=0),0)\\
= &\sum_x\Pr(X_1=x,S=0)\\
=&\Pr(S=0)\\
\leq &\epsilon.
\end{align*}
\end{proof}

The next proposition is a data processing inequality for total variation distance: If you have two random variables, you cannot increase their distance by taking (random) functions of them.

\begin{prop}\label{prop:dpi}
If $X_1$ and $X_2$ have total variation distance $\epsilon$, $Y$ is a random variable independent from $X_1$ and $X_2$ and $f$ is a function. Then the total variation distance between $f(X_1,Y)$ and $f(X_2,Y)$ is at most $\epsilon$.
\end{prop}
\begin{proof}
By Proposition \ref{prop:flag} we can find a distribution of $(X_1,S)$ such that $\Pr(X_1=x,S=1)\leq \Pr(X_2=x)$ and $\Pr(S\neq 1)\leq \epsilon$. As $X_1$ is independent from $Y$, we can have $S$ independent from $Y$. Now $\Pr(f(X_1,Y)=z,S=1)\leq \Pr(f(X_2,Y)=z)$ and $\Pr(S\neq 1)\leq \epsilon$, so $f(X_1,Y)$ and $f(X_2,Y)$ have total variation distance at most $\epsilon$.
\end{proof}

\begin{prop}\label{prop:disjoint}
Let $X_1,\dots, X_n,Y_1,\dots ,Y_n,I$ be independent random variables with $X_i$ and $Y_i$ distributed on $\mathcal{X}_i$, and $I$ distributed on $[n]$. Let $X=X_I$ and $Y=Y_I$. We have
\[\delta(X,Y)\leq \sum_{i=1}^n\Pr(I=i)\delta(X_i,Y_i),\]
with equality if all the $\mathcal{X}_i$'s are pairwise disjoint.

\end{prop}
\begin{proof}
First assume that the $\mathcal{X}_i$'s are disjoint. Then the value of $I$ can be deduced from $X=X_I$ alone and from $Y=Y_I$ alone. This gives us 
\begin{align*}
\delta(X,Y)=& \delta((X,I),(Y,I))\\
=&\sum_{(x,i)}\max(\Pr(X_i=x,I=i)-\Pr(Y_i=x,I=i),0)\\
=&\sum_{(x,i)}\Pr(I=i)\max(\Pr(X_i=x|I=i)-\Pr(Y_i=x|I=i),0)\\
=&\sum_i\Pr(I=i)\delta(X_i,Y_i).
\end{align*}

Without the assumption that the $\mathcal{X}_i$'s are disjoint we have
\begin{align*}
\delta(X,Y)=&\sum_{x}\max\left(\sum_{i}\left(\Pr(X_i=x,I=i)-\Pr(Y_i=x,I=i)\right),0\right)\\
\leq & \sum_x\sum_i\max(\Pr(X_i=x,I=i)-\Pr(Y_i=x,I=i),0)\\
=& \sum_i\Pr(I=i)\delta(X_i,Y_i).
\end{align*}
\end{proof}

\section{Construction}\label{sec:construction}

   In this section we will construct random variables $(X_1,\dots, X_n)$ such that $X_{-i}$ and $X_{-j}$ are $\epsilon$-indistinguishable for all $i$ and $j$. First we consider the case $n=2$.
   
   \begin{prop}\label{prop:n2construction}
   Given $\epsilon$, there exists random variables $(X_1,X_2)$ such that $1\leq X_1<X_2\leq \lceil\frac{1}{\epsilon}\rceil+1$ are integers and $(X_1,X_2)$ has $\epsilon$-indistinguishable $1$-subsets.
   \end{prop}
   \begin{proof}
   Let $X_1$ be uniformly distributed on $\left[\lceil\frac{1}{\epsilon}\rceil\right]$ and let $X_2=X_1+1$. Then 
   \[\delta(X_1,X_2)=\left\lceil\frac{1}{\epsilon}\right\rceil^{-1}\leq \left(\frac{1}{\epsilon}\right)^{-1}=\epsilon\]
   so $(X_1,X_2)$ has $\epsilon$-indistinguishable $1$-subsets.
   \end{proof}
   
   For $n=2$ we only needed to check that $\delta(X_1,X_2)\leq \epsilon$, but for general values of $n$ there are $\binom{n}{2}$ ways of choosing the two $n-1$-subsets. In order to simplify the proofs, we will first argue that it is enough to consider \emph{neighbouring} $n-1$-subsets defined as follows. 

    \begin{defi}
 We say that $(X_1,\dots,X_n)$ \emph{has $\epsilon$-indistinguishable neighbouring $n-1$-subsets} if for any $i\in [n-1]$ the random tuples $X_{-i}$ and $X_{-(i+1)}$ have total variation distance at most $\epsilon$.
\end{defi}

\begin{prop}\label{prop:naboernok}
Fix $n\in \N$ and $\epsilon>0$. If $(X_1,\dots , X_n)$ has $\epsilon$-indistinguishable neighbouring $n-1$-subsets, is has $(n-1)\epsilon$-indistinguishable $n-1$-subsets. 
\end{prop}
\begin{proof} 
This follows from repeated use of the triangle inequality.
\end{proof}

\begin{prop}\label{prop:uniformdelta}
Let $n_1>n_2$ and let $U_{n_1}$ and $U_{n_2}$ be independent random variables uniformly distributed on $[n_1]$ respectively $[n_2]$. Then
\[\delta(U_{n_1},U_{n_2}+U_{n_1})=\frac{n_2+1}{2n_1}.\]
\end{prop}
\begin{proof}
We have
\[\Pr(U_{n_1}=x)=\left\{\begin{array}{ll} \frac{1}{n_1} & \text{, if } x\in \{1,\dots,n_1\}\\ 0 & \text{, otherwise}\end{array}\right. \]
and
\[\Pr(U_{n_1}+U_{n_2}=x)\left\{\begin{array}{ll} \frac{x-1}{n_1n_2} & \text{, if }x\in \{1,\dots, n_2\} \\ \frac{1}{n_1} & \text{, if } x\in \{n_2+1,\dots n_1+1\} \\ \frac{1+n_1+n_2-x}{n_1n_2} & \text{, if } x\in\{n_1+1,\dots n_1+n_2\}\\ 0 & \text{, otherwise.} \end{array}\right.\]
So we get 
\[ \max(\Pr(U_{n_1}=x)-\Pr(U_{n_1}+U_{n_2}=x),0)=\left\{\begin{array}{ll} \frac{n_2-x+1}{n_1n_2} & \text{, if }x\in \{1,\dots, n_2\} \\ 0 & \text{, otherwise}\end{array}\right. .\]
Summing over all $x\in \{1,\dots, n_2\}$ gives us
\begin{align*}
\delta(U_{n_1},U_{n_1}+U_{n_2})=&\sum_x \frac{n_2-x+1}{n_1n_2}\\
=&\frac{n_2}{n_1}-\frac{n_2(n_2+1)}{2n_1n_2}+\frac{n_2}{n_1n_2}\\
=&\frac{n_2+1}{2n_1}.
\end{align*}
\end{proof}

We are now ready for the construction of a distribution of $X=(X_1,\dots, X_n)$ for $n\geq 3$. First we give the construction for $n=3$ and then we construct a distribution for $n$ given a distribution for $n-1$. For this recursive construction to work we need to assume more than just having $\epsilon$-indistinguishable $n-2$-subsets about the distribution for $n-1$ so we cannot use the distribution for $n=2$ as the start of the recursive definition.

\begin{lemm}\label{lemm:construction}
For all $n\geq 3$ and all $\epsilon\in (0,1)$ there exists random variables $X_1,\dots, X_n$ where the $X_i$ takes values in $\N$, with a joint distribution such that $X_1<\dots <X_n$ and
\begin{enumerate}
\item $X$ has $\epsilon(1-2^{-n})$-indistinguishable neighbouring $n-1$-subsets\label{req:close}
\item $\forall i\in[n-1]:\Pr( X_{i+1}-X_{i}< n+4-\log(\epsilon))\leq \epsilon 2^{-n-3}.$\label{req:space}
\item $X_n$ never takes values above $\exp_2^{n-2}\left(4\left\lceil \frac{1}{\epsilon}\right\rceil +6\right)-4n-2+2\log(\epsilon)$.\label{req:bound}
\end{enumerate}
\end{lemm}
\begin{proof}
We fix $\epsilon$ and prove the statement by induction in $n$, so first we show the statement for $n=3$. Let $X_1$ be uniformly distributed on $\left[2^{4\left\lceil\frac{1}{\epsilon}\right\rceil+4}\right]$ and let $K$ be uniformly distributed on $\left[\left\lceil \frac{1}{\epsilon}\right\rceil+3,\dots,4\left\lceil \frac{1}{\epsilon}\right\rceil+3\right]$. We now define $X_2=X_1+2^{K}$ and $X_3=X_2+2^{K}=X_1+2^{K+1}$. 

We see that the only values that $(X_1,X_2)$ can take, but $(X_1,X_3)$ cannot take, are the values where $X_1$ and $X_2$ differ by the smallest possible value, $2^{\left\lceil \frac{1}{\epsilon}\right\rceil+3}$, that is, $K$ is taking the smallest possible value, $\left\lceil \frac{1}{\epsilon}\right\rceil+3$. As $K$ is uniformly distributed on a set with $3\left\lceil \frac{1}{\epsilon}\right\rceil+1$ elements, this happens with probability 
\begin{align*}
\frac{1}{3\left\lceil \frac{1}{\epsilon}\right\rceil+1}\leq \frac{1}{\frac{3}{\epsilon}}=\frac{\epsilon}{3}.
\end{align*}
For all other values $(x_1,x_2)$ of $(X_1,X_2)$, if $K$ was one lower we would have had $(X_1,X_3)=(x_1,x_2)$. As $K$ is uniformly distributed we have
\[\delta((X_1,X_2),(X_1,X_3))\leq \frac{\epsilon}{3}<\frac{7}{8}\epsilon.\]
This shows the $i=2$ case of requirement \ref{req:close}.

To bound $\delta((X_1,X_3),(X_2,X_3))$, we first want bound $\delta((X_1,X_2),(X_2,X_3))$. The only values that can be taken by $(X_1,X_2)$ but not by $(X_2,X_3)$, are values $(x_1,x_2)$ with $x_1\leq \frac{x_2}{2}$, which is equivalent to $x_1\leq 2^K$. If $k=4\left\lceil \frac{1}{\epsilon}\right\rceil+3$, then $\Pr(X_1\leq 2^k)=\frac{1}{2}$, if $k$ is one lower, then $\Pr(X_1\leq 2^k)=\frac{1}{4}$ and so on. In total we get
\begin{align*}
\Pr(X_1\leq 2^K)=&\sum_{k=\left\lceil \frac{1}{\epsilon}\right\rceil+3}^{4\left\lceil \frac{1}{\epsilon}\right\rceil+3]} \Pr(X_1\leq 2^K|K=k)\Pr(K=k)\\
\leq&\left(\frac{1}{2}+\frac{1}{4}+\dots +\right)\frac{1}{3\left\lceil \frac{1}{\epsilon} \right\rceil+1}\\
\leq &\frac{\epsilon}{3}.
\end{align*}
Furthermore, for all values $(x_1,x_2)$ of $(X_1,X_2)$ that $(X_2,X_3)$ can take, we have $\Pr((X_1,X_2)=(x_1,x_2))=\Pr((X_2,X_3)=(x_1,x_2))$ as $(X_1,X_2)$ and $(X_2,X_3)$ are both uniformly distributed on sets of the same sizes. Thus, $\delta((X_1,X_2),(X_2,X_3))\leq \frac{\epsilon}{3}$.
By the triangle inequality we get 
\begin{align*}
\delta((X_1,X_3),(X_2,X_3))=&\delta((X_1,X_3),(X_1,X_2))+\delta((X_1,X_2),(X_2,X_3))\
\leq& \frac{2\epsilon}{3}
\leq & \frac{7}{8}\epsilon,
\end{align*}
showing the remaining case of requirement \ref{req:close}.

Next we want to show requirement $2$.
\begin{align*}
X_3-X_2=&X_2-X_1\\
=&2^K\\
\geq & 2^{\left\lceil \frac{1}{\epsilon}\right\rceil+3}\\
\geq &8\cdot 2^{\frac{1}{\epsilon}}\\
\geq &8\cdot 2^{\log(\frac{1}{\epsilon})}\\
= & 8\cdot \frac{1}{\epsilon}\\
\geq &7+\frac{1}{\epsilon}\\
\geq &7+\log\left(\frac{1}{\epsilon}\right)\\
=&7-\log(\epsilon).
\end{align*}
Here we used $x\geq \log(x)$ twice. This shows that requirement \ref{req:space} holds. 

Finally, we have $X_3=X_1+2^{K+1}\leq 2^{4\left\lceil\frac{1}{\epsilon}\right\rceil+4}+2^{4\left\lceil \frac{1}{\epsilon}\right\rceil+3+1}=2^{4\left\lceil \frac{1}{\epsilon}\right\rceil+5}$. For $\epsilon=1$ this is $2^9\geq 14=4n+2-2\log(\epsilon)$, and $2^{4\left\lceil \frac{1}{\epsilon}\right\rceil+5}\geq 2^{4 \frac{1}{\epsilon}+5}$, which decreases much faster in $\epsilon$ than $14-2\log(\epsilon)$, so we have $2^{4\left\lceil \frac{1}{\epsilon}\right\rceil+5}\geq 14-2\log(\epsilon)$ for all $\epsilon$. This shows that
\[X_3\leq 2^{4\left\lceil \frac{1}{\epsilon}\right\rceil+5}\leq 2^{4\left\lceil \frac{1}{\epsilon}\right\rceil+6}-14+2\log(\epsilon),\]
so requirement \ref{req:bound} is also true.

For the induction step, assume that $(X_1,\dots,X_n)$ satisfy the statement for $n$. We want to construct $(Y_1,\dots, Y_{n+1})$ that shows that the statement holds for $n+1$. To do this we construct a joint distribution of $(X_1,\dots, X_n,D_1,\dots ,D_n,Y_1,\dots, Y_{n+1})$. We choose $(X_1,\dots, X_n)$ so that it satisfy the requirements for $n$, and given these, we let $D_i$ be uniformly distributed on $[2^{X_i+4n-2\log(\epsilon)}]$ and let $Y_1$ be uniformly distributed on $[\exp_2^{n-1}\left(4\left\lceil \frac{1}{\epsilon}\right\rceil +6\right)/2-4n-6+2\log(\epsilon)]$. All these are independent given $(X_1,\dots, X_n)$. We define $Y_{i+1}=Y_i+D_i$ for $i\in [n]$. 

We now check that $(Y_1,\dots ,Y_{n+1})$ satisfy the three requirements. 

First we want to show that if we are given the tuple $(Y_1,D_1,D_2,\dots, D_{i-1},D_{i+1}, \dots ,D_n)$ containing $Y_1$ and all the $D_j$'s except one, $D_i$, then it will not make much of a difference if we add $D_i$ to $D_{i+1}$. That is, we want to bound 
\[\delta((Y_1,D_1,\dots, D_{i-1},D_i+D_{i+1},D_{i+2},\dots, D_n),(Y_1,D_1,\dots, D_{i-1},D_{i+1},D_{i+2},\dots ,D_n)).\]

To do this, we first get from Proposition \ref{prop:uniformdelta} that
\[\delta((D_i+D_{i+1})|_{(X_i,X_{i+1})=(x_i,x_{i+1})},D_{i+1}|_{(X_i,X_{i+1})=(x_i,x_{i+1})})=\frac{\lfloor 2^{x_i+4n-2\log(\epsilon)}\rfloor+1}{2\cdot \lfloor 2^{x_{i+1}+4n-2\log(\epsilon)}\rfloor}\leq 2^{x_i-x_{i+1}}.\]
Now Proposition \ref{prop:disjoint} gives us
\begin{align*}
\delta((&D_i+D_{i+1},X_i,X_{i+1}),(D_{i+1},X_i,X_{i+1}))\\
=&\sum_{(x_i,x_{i+1})}\Pr((X_i,X_{i+1})=(x_i,x_{i+1}))\delta((D_i+D_{i+1})|_{(X_i,X_{i+1})=(x_i,x_{i+1})},D_{i+1}|_{(X_i,X_{i+1})=(x_i,x_{i+1})}).
\end{align*}
From requirement (\ref{req:space}), we know that $\Pr(X_{i+1}-X_i<n+4-\log(\epsilon))\leq \epsilon 2^{-n-3}$. 
When $x_{i+1}-x_{i}<n+4-\log(\epsilon)$ we have
\[\delta((D_i+D_{i+1})|_{(X_i,X_{i+1})=(x_i,x_{i+1})},D_i|_{(X_i,X_{i+1})=(x_i,x_{i+1})})\leq 1\]
as $\delta$ only takes values in $[0,1]$. In all other cases, we have
\[\delta((D_i+D_{i+1})|_{(X_i,X_{i+1})=(x_i,x_{i+1})},D_i|_{(X_i,X_{i+1})=(x_i,x_{i+1})})\leq 2^{-(n+4)+\log(\epsilon)}=\epsilon 2^{-(n+4)}.\]
Summing up gives 
\[\delta((D_i+D_{i+1},X_i,X_{i+1}),(D_{i+1},X_i,X_{i+1}))\leq \epsilon 2^{-n-3}+\epsilon 2^{-n-4}\leq \epsilon 2^{-n-2}.\]

Given $X_i$ and $X_{i+1}$ and either $D_{i}+D_{i+1}$ or $D_{i+1}$, there is a random function giving $(Y_1,D_1,\dots,D_{i-1},D_i+D_{i+1},D_{i+2},\dots, D_n)$ respectively $(Y_1,D_1,\dots,D_{i-1},D_{i+1},D_{i+2},\dots, D_n)$. Thus by Proposition \ref{prop:dpi} we have
\begin{align*}
\delta((Y_1&,D_1,\dots, D_{i-1},D_i+D_{i+1},D_{i+2},\dots, D_n),(Y_1,D_1,\dots, D_{i-1},D_{i+1},D_{i+2},\dots ,D_n))\\
\leq &\delta((D_i+D_{i+1},X_i,X_{i+1}),(D_{i+1},X_i,X_{i+1}))\\
\leq &\epsilon 2^{-n-2}.
\end{align*}
This is the upper bound we wanted.

Clearly there is a random function, not depending on $i$, that given $(X_1,\dots,X_{i-1},X_{i+1},\dots X_n )$ returns $(Y_1,D_1,\dots, D_{i-1},D_{i+1},\dots, D_n)$ such that when input have the correct distribution, then the output have the correct distribution. Thus, 
\begin{align*}
\delta((Y_1&,D_1,\dots ,D_{i-1},D_{i+1},\dots D_n),(Y_1,D_1,\dots,D_i,D_{i+2},\dots D_n))\\
\leq &\delta((X_1,\dots, X_{i-1},X_{i+1},\dots, X_n),(X_1,\dots, X_i,X_{i+2},\dots X_n))\\
\leq& \epsilon(1-2^{-n}).
\end{align*}
For $i\geq 2$ we use the fact that the $Y_j$'s can be computed from the $D_j$'s and $Y_1$ and then use the triangle inequality to get
\begin{align*}
\delta((Y_1&,\dots Y_{i-1},Y_{i+1},\dots ,Y_n),(Y_1,\dots, Y_{i},Y_{i+2},\dots ,Y_n))\\
\leq &\delta((Y_1,D_1,\dots D_{i-2},D_{i-1}+D_{i},D_{i+1},\dots, D_n),(Y_1,D_1,\dots , D_{i-1},D_{i}+D_{i+1},\dots,D_n))\\
\leq & \delta((Y_1,D_1,\dots D_{i-2},D_{i-1}+D_{i},D_{i+1},\dots, D_n),(Y_1,D_1,\dots D_{i-2},D_{i},D_{i+1},\dots, D_n))\\
&+\delta((Y_1,D_1,\dots D_{i-2},D_{i},D_{i+1},\dots, D_n),(Y_1,D_1,\dots , D_{i-2},D_{i-1},D_{i+1},\dots,D_n))\\
&+\delta((Y_1,D_1,\dots , D_{i-2},D_{i-1},D_{i+1},\dots,D_n),(Y_1,D_1,\dots , D_{i-1},D_{i}+D_{i+1},\dots,D_n))\\
\leq & 2\cdot \epsilon 2^{-n-2}+\epsilon(1-2^{-n})\\
=& \epsilon (1-2^{-(n+1)}).
\end{align*}
 This shows requirement (\ref{req:close}) in the case $i\geq 2$. 

Similarly, we want to bound
\[\delta((Y_1,D_2,D_3,\dots,D_n),(Y_1+D_1,D_2,D_3,\dots, D_n)).\]
$Y_1$ is chosen uniformly from $[\frac{\exp_2^{n-1}\left(4\left\lceil \frac{1}{\epsilon}\right\rceil +6\right)}{2}-4n-6+2\log(\epsilon)]$ and $D_1$ uniformly from $[2^{X_1+4n-2\log(\epsilon)}]$.

As $X_1<\dots<X_n$ and $X_n\leq \exp_2^{n-2}\left(4\left\lceil \frac{1}{\epsilon}\right\rceil +6\right)-4n-2+2\log(\epsilon)$ we get $X_1\leq \exp_2^{n-2}\left(4\left\lceil \frac{1}{\epsilon}\right\rceil +6\right)-5n-1+2\log(\epsilon)$.
From Proposition \ref{prop:uniformdelta} we get
\begin{align*}
\delta(Y_1|_{X_1=x_1},(Y_1+D_1)|_{X_1=x_1})=&\frac{\left\lfloor 2^{x_1+4n-\log(\epsilon)}\right\rfloor+1}{2\left\lfloor \frac{\exp_2^{n-1}\left(4\left\lceil \frac{1}{\epsilon}\right\rceil +6\right)}{2}-4n-6+2\log(\epsilon)\right\rfloor}\\
\leq & \frac{2^{\exp_2^{n-2}\left(4\left\lceil \frac{1}{\epsilon}\right\rceil +6\right)-n-1+\log(\epsilon))} +1}{2^{\exp_2^{n-2}\left(4\left\lceil \frac{1}{\epsilon}\right\rceil +6\right)} -8n-14+4\log(\epsilon)}\\
\leq & \frac{ 2^{\exp_2^{n-2}\left(4\left\lceil \frac{1}{\epsilon}\right\rceil +6\right)-n-1+\log(\epsilon))} }{2^{\exp_2^{n-2}\left(4\left\lceil \frac{1}{\epsilon}\right\rceil +6\right)-1}}\\
\leq& 2^{-n+\log(\epsilon)}\\
=&\epsilon 2^{-n}.
\end{align*}
In the second inequality we used that $2^{\exp_2^{n-2}\left(4\left\lceil \frac{1}{\epsilon}\right\rceil +6\right)-1}\geq 8n+14-4\log(\epsilon)$.
Summing up over all possible $x_1$ we get
\begin{align*}
\delta((Y_1,X_1),(Y_1+D_1,X_1))=&\sum_{x_1}\Pr(X_1=x_1)\delta(Y_1|_{X_1=x_1},(Y_1+D_1)|_{X_1=x_1})\\
\leq &\epsilon 2^{-n}.
\end{align*}
Similarly to before, this implies 
\[\delta((Y_1,D_2,D_3,\dots, D_n),(Y_1+D_1,D_2,D_3,\dots,D_n))\leq \epsilon 2^{-n}\]
and hence
\begin{align*}
\delta((Y_1&,Y_3,Y_4,\dots, Y_{n+1}),(Y_2,Y_3,Y_4,\dots, Y_{n+1}))\\
\leq & \delta((Y_1,D_1+D_2,D_3,\dots, D_n),(Y_1+D_1,D_2,D_3,\dots, D_n))\\
\leq & \delta((Y_1,D_1+D_2,D_3,\dots,D_n),(Y_1,D_2,D_3,\dots,D_n))\\
&+\delta((Y_1,D_2,D_3,\dots,D_n),(Y_1+D_1,D_2,D_3,\dots, D_n))\\
\leq & \epsilon 2^{-n-2}+\epsilon 2^{-n}\\
\leq &\epsilon(1-2^{-(n+1)}).
\end{align*}
This shows the remaining case, $i=1$, of requirement \ref{req:close}.

We have $Y_{i+1}-Y_{i}=D_i$, so requirement \ref{req:space} can be written as
\[\Pr(D_i<(n+1)+4-\log(\epsilon))\leq \epsilon 2^{-(n+1)-3}.\]
By definition of $D_i$ we have
\begin{align*}
\Pr(D_i< n+5-\log(\epsilon))\leq& (n+5-\log(\epsilon)) 2^{-4n+2\log(\epsilon)-X_i}\\
\leq & \frac{ n+5-\log(\epsilon)}{2^{4n-2\log(\epsilon)+1}}\\
= & \epsilon\frac{ n+5-\log(\epsilon)}{2^{4n-\log(\epsilon)+1}}\\
\leq & \epsilon 2^{-(n+1)-3}.
\end{align*}
To see the last inequality in the above computation, we can rewrite it as $ n+5-\log(\epsilon)\leq 2^{-(n+1)-3}2^{4n-\log(\epsilon)+1}=2^{3n-3}2^{-\log(\epsilon)}$. For $n\geq 3$ and $\epsilon=1$ it is easy to see that is holds. To generalise this to all $\epsilon$ we write $t=-\log(\epsilon)$, and see that the right hand side increases faster in $t$ than the left hand side. This shows that requirement (\ref{req:space}) holds.

By construction, $Y_{n+1}=Y_1+\sum_{i=1}^nD_i$, 
where $Y_1\leq \exp_2^{n-1}\left(4\left\lceil \frac{1}{\epsilon}\right\rceil +6\right)/2-4n-6+2\log(\epsilon)$, 
and $D_i\leq 2^{X_i+4n-2\log(\epsilon)}$. 
As $X_n\leq \exp_2^{n-2}\left(4\left\lceil \frac{1}{\epsilon}\right\rceil +6\right)-4n-2+2\log(\epsilon)$ we 
get $D_n\leq 2^{\exp_2^{n-2}\left(4\left\lceil \frac{1}{\epsilon}\right\rceil +6\right)-2}=\frac{\exp_2^{n-1}\left(4\left\lceil \frac{1}{\epsilon}\right\rceil +6\right)}{4}$. We know that the $X_i$'s are increasing so $D_{n-1}$ is at most half this size and so on. In total we get 
\begin{align*}
Y_{n+1}\leq & \left(2^{-1}+2^{-2}+\dots+2^{-n}\right)\exp_2^{n-1}\left(4\left\lceil \frac{1}{\epsilon}\right\rceil +6\right)-4n-6+2\log(\epsilon)\\
\leq &\exp_2^{(n+1)-2}\left(4\left\lceil \frac{1}{\epsilon}\right\rceil +6\right)-4(n+1)-2+2\log(\epsilon).
\end{align*}
This proves that the last requirement holds.
\end{proof}

\begin{coro}\label{coro:upperbound}
For fixed $n$ there exists a distribution of $X=(X_1,\dots, X_n)$ where $1\leq X_1<X_2<\dots<X_n\leq N(\epsilon)$ are all integers and $X$ has $\epsilon$-indistinguishable $n-1$-subsets and $N(\epsilon)=\exp_2^{n-2}(O(\frac{1}{\epsilon}))$.
\end{coro}
\begin{proof}
Follows from Proposition \ref{prop:n2construction} and \ref{prop:naboernok} and Lemma \ref{lemm:construction}.
\end{proof}

\begin{coro}
For fixed $n$ the value $v_1$ of Game $1$ as a function of $N$ is $O\left(\frac{1}{\log^{n-2}(N)}\right)$
\end{coro}
\begin{proof}
This follows from Theorem \ref{theo:equivalence} and Corollary \ref{coro:upperbound}. 
\end{proof}

\begin{coro}
For fixed $n$ there exists a distribution of $X=(X_1,\dots, X_n)$ where $1\leq X_1<X_2<\dots<X_n\leq N(\epsilon)$ are all integers and $X$ has $\epsilon$-indistinguishable subsets and $N(\epsilon)=\exp_2^{n-2}(O(\frac{1}{\epsilon}))$.
\end{coro}
\begin{proof}
Follows from Proposition \ref{prop:subsetfromnminus1} and Corollary \ref{coro:upperbound}.
\end{proof}

\section{Lower bounds}\label{sec:lower bound}

In this section we will show lower bounds on how large values $X_n$ need to take if $X=(X_1,\dots, X_n)$ has $\epsilon$-indistinguishable $n-1$-subsets and we always have $X_1\geq 0$ and $X_{i+1}\geq X_i+1$. We no longer require that the $X_i$ are integers, only that there are at least one apart. This weaker requirement makes the induction argument easier. Clearly, any lower bound we show under the assumption that the $X_i$'s are at least one apart will also be a lower bound in the case where the $X_i$ have to take integer values. Conversely, if you have a distribution of $X$ with $\epsilon$-indistinguishable $n-1$-subsets and $X_1\geq 0, X_{i+1}\geq X_i+1$ you can define $X'$ by $X'_i=1+\lfloor X_{i}\rfloor$. Then $X_1<X_2<\dots<X_n$ will be natural numbers and by Proposition \ref{prop:dpi} $X'$ will have $\epsilon$-indistinguishable $n-1$-subsets. 

\begin{prop}\label{prop:tooln2}
If $X_1$ and $X_2$ are discrete random variables taking real values in an interval $[a,b]$ and $\E X_2\geq \E X_1+1$ then 
\[\delta(X_1,X_2)\geq \frac{1}{b-a}.\] 
\end{prop}
\begin{proof}
If $a\neq 0$ we can subtract $a$ from $X_1$ and $X_2$, and set the new $b$ to be $b-a$ and $a$ to be $0$. We will still have $\E X_2\geq \E X_1+1$ and the distance $\delta(X_1,X_2)$ and $b-a$ are not affected by this. So in the following we will assume $a=0$.

Then we have
\begin{align*}
\E X_2 -\E X_1 =& \sum_{x}x(\Pr(X_2=x)-\Pr(X_1=x))\\
\leq &\sum_x x\max(\Pr(X_2=x)-\Pr(X_1=x),0) \\
\leq &b\sum_x  \max(\Pr(X_2=x)-\Pr(X_1=x),0)\\
=&b\delta(X_2,X_1)\\
=&b\delta(X_1,X_2).
\end{align*}
So $\delta(X_1,X_2)\geq \frac{1}{b}=\frac{1}{b-a}$.
\end{proof}

We can now show a lower bound in the case $n=2$.

\begin{prop}\label{prop:uppern2}
If $X_1,X_2$ are random variables over the non-negative real numbers such that $X_2\geq X_1+1$ and $(X_1,X_2)$ has $\epsilon$-indistinguishable $1$-subsets, then $X_2$ must take values of at least $\frac{1}{\epsilon}$ with positive probability. 
\end{prop}
\begin{proof}
To say that $(X_1,X_2)$ has $\epsilon$-indistinguishable $1$-subset means that $\delta(X_2,X_1)\leq \epsilon$. As $X_2\geq X_1+1$ we have $\E X_2\geq \E X_1+1$ and the statement follows from Proposition \ref{prop:tooln2}.
\end{proof}

Here we allowed $X_1$ to be $0$. If we required $X_1$ and $X_2$ to be natural numbers, the lower bound would be $\lceil \frac{1}{\epsilon}\rceil+1$, which exactly matches our construction in Proposition \ref{prop:n2construction}. 

We can combine Proposition \ref{prop:tooln2} with Proposition \ref{prop:dpi} to get the following.

\begin{prop}\label{prop:tooln3}
If $X_1$ and $X_2$ are random variables with domain $\mathcal{X}$ and $f:\mathcal{X}\to [a,b]$ is a function such that $\E f(X_2)\geq \E f(X_1)+1$ then $\delta(X_1,X_2)\geq \frac{1}{b-a}$.
\end{prop}
\begin{proof}
By Proposition \ref{prop:tooln2} we have $\delta(f(X_1),f(X_2))\geq \frac{1}{b-a}$, so by Proposition \ref{prop:dpi} we get $\delta(X_1,X_2)\geq \frac{1}{b-a}$.
\end{proof}

We will now show the lower bound in the case $n=3$.

\begin{prop}\label{prop:uppern3}
Let $X_1,X_2,X_3$ be random variables taking non-negative real numbers such that $X_2\geq X_1+1$ and $X_3\geq X_2+1$. If $(X_1,X_2,X_3)$ has $\epsilon$-indistinguishable $2$-subsets, then $X_3$ must take values of at least $2^{\frac{1}{\epsilon}}$ with positive probability. 
\end{prop}
\begin{proof}
Let $f(x,y)=\log(y-x)$. We have $X_3-X_1=(X_3-X_2)+(X_2-X_1)$ so by Jensen's inequality we get
\[\log(X_3-X_1)-1=\log\left(\frac{X_3-X_1}{2}\right) \geq \frac{\log(X_3-X_2)+\log(X_2-X_1)}{2}.\]
Thus
\[\E 2\log(X_3-X_1)-2\geq \E \log(X_3-X_2)+\E \log(X_2-X_1).\]
We must have at least one of $\E \log(X_3-X_1)\geq \E \log(X_3-X_2)+1$ and $\E \log(X_3-X_1)\geq \E \log(X_2-X_1)+1$. Assume without loss of generality that the first one is the case. As $(X_1,X_2,X_3)$ has $\epsilon$-indistinguishable $2$-subsets we have $\delta((X_1,X_3),(X_2,X_3))\leq 
\epsilon$ so Proposition \ref{prop:tooln2} tell us that the $\log$'s must take values in an interval of length $\frac{1}{\epsilon}$. As the $X_i$'s always differ by at least one, the $\log$'s only take non-negative values. Hence, $\log(X_3-X_1)\geq \frac{1}{\epsilon}$ with positive probability, so $X_3\geq X_3-X_1\geq 2^{\frac{1}{\epsilon}}$ with positive probability. 
\end{proof}

In later proofs we would like to be able to ignore events that only happen with small probability, and argue that this does not increase the total variation distance between two random variables too much. In order to do that, we need the following proposition.

\begin{prop}\label{prop:givenT}
Let $(X_1,X_2,T)$ be random variables with some joint distribution, where $T$ only takes values $0$ and $1$ and $\Pr(T=0)=\epsilon<1$ and $\delta(X_1,X_2)=\delta$. For $i\in\{1,2\}$ define $X_i'=X_i|_{T=1}$. Then 
\[\delta(X_1',X_2')\leq \frac{\delta+\epsilon}{1-\epsilon}.\]
\end{prop}
\begin{proof}
\begin{align*}
\delta(X_1',X_2')=& \sum_x \max(\Pr(X_1'=x)-\Pr(X_2'=x),0)\\
=& \sum_x \max(\Pr(X_1=x|T=1)-\Pr(X_2=x|T=1),0)\\
=& \sum_x \frac{\max(\Pr(X_1=x,T=1)-\Pr(X_2=x,T=1),0)}{\Pr(T=1)}\\
\leq & \sum_x \frac{\max(\Pr(X_1=x)-\Pr(X_2=x,T=1),0)}{\Pr(T=1)}\\
=&\sum_x \frac{\max(\Pr(X_1=x)-\Pr(X_2=x)+\Pr(X_2=x)-\Pr(X_2=x,T=1),0)}{\Pr(T=1)}\\
\leq &\sum_x \frac{\max(\Pr(X_1=x)-\Pr(X_2=x),0)+\max(\Pr(X_2=x)-\Pr(X_2=x,T=1),0)}{\Pr(T=1)}\\
\leq&\frac{\delta+\epsilon}{1-\epsilon}.
\end{align*}
\end{proof}

We will now consider the case $n=4$. Before we show the lower bound, we will show that if $X$ has $\epsilon$-indistinguishable $3$-subsets then $(X_1,X_2,X_3,X_4)$ will with high probability be in one of two cases. Intuitively, one of these cases corresponds to the gaps $X_2-X_1, X_3-X_2,X_4-X_3$ increasing and the other corresponds to the gaps decreasing.  

\begin{prop}\label{prop:n4isbiased}
Let $X_1,X_2,X_3,X_4$ be discrete random variable taking real values such that $X_1<X_2<X_3<X_4$. Assume that $(X_1,X_2,X_3,X_4)$ has $\epsilon$-indistinguishable $3$-subsets. Then with probability at least $1-9\epsilon$ we have $\frac{X_3-X_2}{X_4-X_1}< \frac{1}{4}$ and one of the following
\begin{enumerate}
\item$X_3< \frac{X_1+X_4}{2}$ and $X_2\leq \frac{X_1+X_3}{2}$, or
\item $X_2> \frac{X_1+X_4}{2}$ and $X_3> \frac{X_2+X_4}{2}$.
\end{enumerate}

\end{prop}

\begin{proof}
Define $f(x,y,z)=\frac{y-x}{z-x}\epsilon^{-1}$. We see that $0< f(X_1,X_2,X_4)< f(X_1,X_3,X_4)< \epsilon^{-1}$. As $\delta((X_1,X_2,X_4),(X_1,X_3,X_4))\leq \epsilon$, Proposition \ref{prop:tooln2} implies that 
\[\E\left( f(X_1,X_3,X_4)- f(X_1,X_2,X_4)\right)< 1.\]
That is $\E \frac{X_3-X_2}{X_4-X_1}< \epsilon$. In particular 
\[\Pr\left(\frac{X_3-X_2}{X_4-X_1}\geq\frac{1}{4}\right)< 4 \epsilon.\]
 Define $T$ to be the random variable that is $T=1$ when $\frac{X_3-X_2}{X_4-X_1}<\frac{1}{4}$ and otherwise is $0$. Let $(X_1',X_2',X_3',X_4')=(X_1,X_2,X_3,X_4)|_{T=1}$. Now $(X_1',X_2',X_3',X_4')$ has $\frac{\Pr(T=0)+\epsilon}{\Pr(T=1)}$-indistinguishable $3$-subset: to see for example that $\delta((X_1',X_2',X_3'),(X_1',X_2',X_4'))\leq\frac{\Pr(T=0)+\epsilon}{\Pr(T=1)}$ we use Proposition \ref{prop:givenT} on $((X_1,X_2,X_3),(X_1,X_2,X_4),T)$, and similar for all other pairs of $3$-subsets. 
 
 Now define $g(x,y,z)=1$ if $y> \frac{x+z}{2}$ and otherwise $g(x,y,z)=0$. As $4(X_3'-X_2')\leq X_4'-X_1'$ we have
 \[g(X_1',X_2',X_3')\geq g(X_1',X_3',X_4')\geq g(X_1',X_2',X_4')\geq g(X_2',X_3',X_4').\]
 Here the middle inequality follows from $X_3'>X_2'$. To show the first inequality, assume for contradiction that it is wrong for some particular values $x_1,x_2,x_3,x_4$ of $X_1',X_2',X_3',X_4'$. Then we must have $g(x_1,x_3,x_4)=1$, so $x_3>\frac{x_1+x_4}{2}$. But that implies $\frac{x_3-x_1}{x_4-x_1}>\frac{1}{2}$ and as $\frac{x_3-x_2}{x_4-x_1}<\frac{1}{4}$ this implies $x_2>\frac{x_1+x_3}{2}$ and the first inequality is true. The last inequality is similar. 
 
 By proposition \ref{prop:dpi} we know that $\delta(g(X_1',X_2',X_3'),g(X_2',X_3',X_4'))\leq \frac{\Pr(T=0)+\epsilon}{\Pr(T=1)}$, so $\E g(X_1',X_2',X_3')-g(X_2',X_3',X_4')\leq  \frac{\Pr(T=0)+\epsilon}{\Pr(T=1)}$. Because $g$ only takes the values $0$ and $1$ and $g(X_1',X_2',X_3')\geq g(X_2',X_3',X_4')$ we have
 \[\Pr(g(X_1',X_2',X_3')\neq g(X_2',X_3',X_4'))\leq  \frac{\Pr(T=0)+\epsilon}{\Pr(T=1)}.\]
 Let $T'$ be the random variable that is $0$ when $T=0$ or $g(X_1',X_2',X_3')\neq g(X_2',X_3',X_4')$. We have
 \begin{align*}
 \Pr(T'=0)=& \Pr(T=0)+\Pr(T=1, g(X_1,X_2,X_3)\neq g(X_2,X_3,X_4))\\
 \leq & \Pr(T=0)+\Pr(T=1) \frac{\Pr(T=0)+\epsilon}{\Pr(T=1)}\\
 \leq & 4\epsilon +4\epsilon+\epsilon\\
 =&9\epsilon.
  \end{align*}
  If $g(X_1,X_2,X_3)= g(X_2,X_3,X_4)=0$ we have $X_3\leq \frac{X_2+X_4}{2}< \frac{X_1+X_4}{2}$ and $X_2\leq \frac{X_1+X_3}{2}$ and we are in the first case of the conclusion of the proposition. Similarly, if $g(X_1,X_2,X_3)= g(X_2,X_3,X_4)=1$ we have $X_2>\frac{X_1+X_3}{2}>\frac{X_1+X_4}{2}$ and $X_3>\frac{X_2+X_4}{2}$.
\end{proof}

We are now ready to show the lower bound in the case $n=4$.

\begin{prop}\label{prop:n4}
Let $X_1,X_2,X_3,X_4$ be random variables over the non-negative real numbers such that $X_{i+1}\geq X_i+1$ for $i\in\{1,2,3\}$ and let $\epsilon<\frac{1}{9}$. If $(X_1,X_2,X_3,X_4)$ has $\epsilon$-indistinguishable $3$-subsets, then $X_4$ must take values of at least $\exp_2^2\left(\frac{1-9\epsilon}{20\epsilon}\right)$ with positive probability. 
\end{prop}
\begin{proof}
First we consider the case where we always have $X_3\leq \frac{X_1+X_4}{2}$ and $X_2\leq \frac{X_1+X_3}{2}$. From this we conclude that $X_4-X_1\geq 2(X_3-X_1)\geq 4(X_2-X_1)$. In other words $\log(X_4-X_1)\geq \log(X_3-X_1)+1$ and $\log(X_3-X_1)\geq \log(X_2-X_1)+1$. We claim that if $(X_1,X_2,X_3,X_4)$ has $\epsilon$-indistinguishabl $3$-subsets, then $(\log(X_2-X_1),\log(X_3-X_1),\log(X_4-X_1))$ has $\epsilon$-indistinguishable $2$-subsets. To show for example that $\delta((\log(X_3-X_1),\log(X_4-X_1)),(\log(X_2-X_1),\log(X_3-X_1)))\leq \epsilon$ we define $f(x,y,z)=(\log(y-x),\log(z-x))$ and use Proposition \ref{prop:dpi} together with the assumption that $\delta((X_1,X_3,X_4),(X_1,X_2,X_3))\leq \epsilon$. Similar for all other pair of $2$-subsets of $\{\log(X_2-X_1),\log(X_3-X_1),\log(X_4-X_1)\}$. As the $X_i$'s differ by one, the $\log$'s are always non-negative, and we have shown that they differ by one. Hence, by Proposition \ref{prop:uppern3} $\log(X_4-X_1)$ most take values of at least $2^{\frac{1}{\epsilon}}$ with positive probability. Thus, $X_4$ must take values of at least $\exp_2^2(\frac{1}{\epsilon})$.

This was assuming $X_3\leq \frac{X_1+X_4}{2}$ and $X_2\leq \frac{X_1+X_3}{2}$. If we instead assume $X_2\geq \frac{X_1+X_4}{2}$ and $X_3\geq \frac{X_2+X_4}{2}$ we can look at $\log(X_4-X_3), \log(X_4-X_2)$ and $\log(X_4-X_1)$, and get the same result. 

Next, suppose that we are only promised that for each value of $(X_1,X_2,X_3,X_4)$ we are in one of those cases, but that it is not always the same of the two cases. Let $I$ be a random variable that is $1$ when we are in the case where the gaps increase and $0$ in the case where the gaps decrease. Given three of $X_1,X_2,X_3,X_4$ we can see which case we are in, even if we do not know which three of them we were given: we simply plug the three numbers into the function $g$ from the proof of Proposition \ref{prop:n4isbiased}. Proposition \ref{prop:disjoint} gives us
\begin{align}
\delta((X_1,X_2,X_3),(X_2,X_3,X_4))=&\sum_{i=0}^1\Pr(I=i)\delta((X_1,X_2,X_3)|_{I=i},(X_2,X_3,X_4)|_{I=i}).
\end{align}
And similar for all other pairs for $3$-subsets. There must be an $i_0$ such that $\Pr(I=i_0)\geq \frac{1}{2}$, and if $(X_1,X_2,X_3,X_4)$ has $\epsilon$-indistinguishable $3$-subsets, then $(X_1,X_2,X_3,X_4)|_{I=i_0}$ must have $2\epsilon$-indistinguishable $3$-subsets, and hence $X_4$ must takes some value of at least $\exp_2^2(\frac{1}{2\epsilon})$ with positive probability.

Finally, without any promises on $X_1,X_2,X_3,X_4$ we know from Proposition \ref{prop:n4isbiased} that with probability $1-9\epsilon$ one of the two requirement holds. Let $T$ be a random variable that is $1$ when one of these holds a $0$ otherwise. Define $(X_1',X_2',X_3',X_4')=(X_1,X_2,X_3,X_4)|_{T=1}$. Using Proposition \ref{prop:givenT} we can show that $(X_1',X_2',X_3',X_4')$ has $\frac{10\epsilon}{1-9\epsilon}$-indistinguishable $3$-subsets. As $X_1',X_2',X_3',X_4'$ always satisfy one of the two requirements, $X_4'$ (and hence $X_4$) must take values of at least $\exp_2^2\left(\left(2\frac{10\epsilon}{1-9\epsilon}\right)^{-1}\right)=\exp_2^2\left(\frac{1-9\epsilon}{20\epsilon}\right)$. 
\end{proof}

In the proof of a lower bound for general $n$, we can use Proposition \ref{prop:n4isbiased} to argue that any four consecutive $X_i$ will either have increasing or decreasing gaps. We can then use the following proposition to argue that all the gaps must be either increasing or decreasing.

\begin{prop}\label{prop:allthesame}
Let $x_1<\dots <x_n$ be a sequence such that for all $i\in [n-3]$ we have $\frac{x_{i+2}-x_{i+1}}{x_{i+3}-x_i}<\frac{1}{4}$ and one of the following two conditions holds:
\begin{enumerate}
\item $x_{i+2}<\frac{x_i+x_{i+3}}{2}$ and $x_{i+1}\leq \frac{x_i+x_{i+2}}{2}$, or
\item $x_{i+1}> \frac{x_i+x_{i+3}}{2}$ and $x_{i+2}> \frac{x_{i+1}+x_{i+3}}{2}$.
\end{enumerate}
Then it must be the same of the two conditions that holds for every $i$. If it is the first then $x_{i+1}-x_i\geq x_i-x_1$ for all $i\in\{2,\dots, n-1\}$. If it is the second then $x_i-x_{i-1}\geq x_n-x_i$ for all $i\in\{2,\dots n-1\}$.
\end{prop}
\begin{proof}
Assume that $x_1,\dots, x_n$ satisfy the condition in the proposition. Consider an $i\in [n-3]$. If we are in the first case we have $2x_{i+1}\leq x_i+x_{i+2}$ so $x_{i+1}-x_i\leq x_{i+2}-x_{i+1}$ and $2x_{i+2}< x_i+x_{i+3}$ so $x_{i+2}-x_{i+1} < x_{i+2}-x_i< x_{i+3}-x_{i+2}$. In other words, the gaps between the $x_i$ do not get smaller. By a similar proof, in the second case the gaps get strictly smaller. Thus, by looking at the gaps $x_{i+2}-x_{i+1}$ and $x_{i+3}-x_{i+2}$ we see that is must be the same case that is true for $i$ and for $i+1$. By induction is must be the same case for all $i$. 

Assume that we are in the first case for all $i$. Then for $i=1$ we already have $x_2-x_1\leq x_3-x_2$. Next assume for induction that $x_{i+1}-x_{i}\geq x_{i}-x_1$ for $i\geq 2$. If we insert $i-1$ instead of $i$ in $x_{i+2}<\frac{x_i+x_{i+3}}{2}$ we get $x_{i+1}<\frac{x_{i-1}+x_{i+2}}{2}$. This is equivalent to $\frac{x_{i+2}-x_{i+1}}{x_{i+2}-x_{i-1}}>\frac{1}{2}$. By inserting $i-1$ instead of $i$ in $\frac{x_{i+2}-x_{i+1}}{x_{i+3}-x_i}\leq \frac{1}{4}$ we get $\frac{x_{i+1}-x_{i}}{x_{i+2}-x_{i-1}}\leq \frac{1}{4}$. Combining these two we get
\begin{align*}
\frac{x_{i+2}-x_{i+1}}{x_{i+1}-x_{i}}=\frac{x_{i+2}-x_{i+1}}{x_{i+2}-x_{i-1}}\cdot \frac{x_{i+2}-x_{i-1}}{x_{i+1}-x_{i}}\geq \frac{1}{2}\cdot\frac{4}{1}=2.
\end{align*}
Thus $x_{i+2}-x_{i+1}\geq 2(x_{i+1}-x_i)\geq x_{i+1}-x_i+x_i-x_1=x_{i+1}-x_1$. Here the last inequality follows from the induction hypothesis. The case where the gaps gets smaller is similar. 
\end{proof}

Finally, we show the lower bound for general $n$.

\begin{theo}\label{theo:upper}
Let $n\geq 4,\epsilon<\left(18^{n-3}(n-2)!\right)^{-1}$ and let $X_1,\dots , X_n$ be random variables over the non-negative real numbers such that $X_{i+1}\geq X_i+1$ for $i\in [n-1]$. If $(X_1,\dots,X_n)$ has $\epsilon$-indistinguishable $n-1$-subsets, then $X_n$ must take values of at least $\exp_2^{n-2}\left(\left(18^{n-3}(n-2)!\epsilon\right)^{-1}\right)$ with positive probability.  
\end{theo}
\begin{proof}
We show this by induction on $n$. The case $n=4$ we know from Proposition \ref{prop:n4} that $X_4$ must take values of at least $\exp_2^2\left(\frac{1-9\epsilon}{20\epsilon}\right)$. For $\epsilon<\frac{1}{18^{n-3}(n-2)!}=\frac{1}{36}$ we clearly have $\exp_2^2\left(\frac{1-9\epsilon}{20\epsilon}\right)\geq \exp_2^{n-2}\left(\left(18^{n-3}(n-2)!\epsilon\right)^{-1}\right)$.

Assume for induction that the theorem is true for $n-1$. For each $i\in [n-3]$ consider $X_i,X_{i+1},X_{i+2},X_{i+3}$. If $(X_1,\dots X_n)$ has $\epsilon$-indistinguishable $n-1$ subsets, then $(X_i,X_{i+1},X_{i+2},X_{i+3})$ has $\epsilon$-indistinguishable $3$-subsets: for example to show that $\delta((X_{i},X_{i+1},X_{i+2}),(X_{i+1},X_{i+2},X_{i+3}))<\epsilon$ we define $f$ to be the function that given an $n-1$-tuple returns the $i$'th, $i+1$'th and $i+2$'th element and use Proposition \ref{prop:dpi} together with the assumption that $\delta(X_{-(i+3)},X_{-i})\leq \epsilon$. 

Now define $T_i$ to be $1$ if $\frac{X_{i+2}-X_{i+1}}{X_{i+3}-X_i}< \frac{1}{4}$ and 
\[X_{i+2}< \frac{X_i+X_{i+3}}{2}\text{ and }X_{i+1}\leq \frac{X_i+X_{i+2}}{2}\]
we define $T_i$ to be $2$ if $\frac{X_{i+2}-X_{i+1}}{X_{i+3}-X_i}< \frac{1}{4}$ and
\[X_{i+1}> \frac{X_i+X_{i+3}}{2}\text{ and }X_{i+2}> \frac{X_{i+1}+X_{i+3}}{2}\]
and we define $T_i=0$ otherwise. By Proposition \ref{prop:n4isbiased}, $\Pr(T_i=0)\leq 9\epsilon$. We define $T$ to be $1$ if all the $T_i$'s are $1$, we define it to be $2$ if all the $T_i$'s are $2$ and we define $T=0$ otherwise. We know from Proposition \ref{prop:allthesame} that ``otherwise'' only happens if one of the $T_i$ are $0$. So by Proposition \ref{prop:n4isbiased} and the union bound $\Pr(T=0)=\Pr(\exists i: T_i=0)\leq (n-3)9\epsilon<1$. Define $(X_1',\dots ,X_n')=(X_1,\dots ,X_n)|_{T\neq 0}$. From Proposition \ref{prop:givenT} we conclude that $(X_1',\dots, X_n')$ has $\frac{\epsilon+(n-3)9\epsilon}{1-(n-3)9\epsilon}$-indistinguishable $n-1$-subsets. We must have $\Pr(T=t|T\neq 0)\geq \frac{1}{2}$ for some $t\in\{1,2\}$. In the following we will assume that this is the case for $t=1$, the proof for $t=2$ is very similar. By the same argument as in the proof of Proposition \ref{prop:n4} we see that, $(X'_1,\dots,X'_n)|_{T=1}$ has $2\frac{\epsilon+(n-3)9\epsilon}{1-(n-3)9\epsilon}$-indistinguishable $n-1$-subsets, and as $\epsilon<\left(18^{n-3}(n-2)!\right)^{-1}<\frac{8}{81(n-2)(n-3)}$ a computation shows that we have $2\frac{\epsilon+(n-3)9\epsilon}{1-(n-3)9\epsilon}\leq 18(n-2)\epsilon$ so $(X'_1,\dots,X'_n)|_{T=1}$ has $(18(n-2)\epsilon)$-indistinguishable $n-1$-subsets. As consecutive $X_i'$'s always differ by at least $1$, the random variables $\log(X_2'-X_1'),\log(X_3'-X_1'), \dots, \log(X_n'-X_1')$ take non-negative values, and as $T=1$ we know from Proposition \ref{prop:allthesame} that consecutive $\log(X_i'-X_1')$'s always differ by at least $1$. We have \begin{align*} 
18(n-2)\epsilon<&\frac{18(n-2)}{18^{n-3}(n-2)!}\\ 
=&\frac{1}{18^{(n-1)-3}((n-1)-2)!},\end{align*}
so by the induction hypothesis $\log(X_n'-X_1')$ must take values of at least 
\[\exp_2^{n-3}\left(\left(18^{(n-1)-3}((n-1)-2)!\left(18(n-2)\epsilon\right)\right)^{-1}\right)=\exp_2^{n-3}\left(\left(18^{n-3}(n-2)!\epsilon\right)^{-1}\right)\] 
so $X_n'$, and hence $X_n$, must take values of at least $\exp_2^{n-2}\left(\left(18^{n-3}(n-2)!\epsilon\right)^{-1}\right)$. 

In the case $\Pr(T=2|T\neq 0)>\frac{1}{2}$ we consider $\log(X_n-X_{n-1}), \log(X_n-X_{n-2}),\dots, \log(X_n-X_1)$ instead of $\log(X_2-X_1),\log(X_3-X_1),\dots, \log(X_n-X_1)$ but otherwise the proof is the same. 
\end{proof}

\begin{coro}\label{coro:Theta}
For fixed $n$ there exists a distribution of $X=(X_1,\dots, X_n)$ where $1\leq X_1<X_2<\dots<X_n\leq N(\epsilon)$ are all integers and $X$ has $\epsilon$-indistinguishable $n-1$-subsets and $N(\epsilon)=\exp_2^{n-2}(\Theta(\frac{1}{\epsilon}))$.
\end{coro}
\begin{proof}
The upper bound was already shown in \ref{coro:upperbound}. The lower bound follows from Proposition \ref{prop:uppern2} and \ref{prop:uppern3} and Theorem \ref{theo:upper}.
\end{proof}

\begin{coro}
For fixed $n$ the value $v_1$ of Game $1$ as a function of $N$ is $\Theta\left(\frac{1}{\log^{n-2}(N)}\right)$
\end{coro}
\begin{proof}
This is follows from Theorem \ref{theo:equivalence} and Corollary \ref{coro:Theta}.
\end{proof}

\begin{coro}
For fixed $n$ there exists a distribution of $X=(X_1,\dots, X_n)$ where $1\leq X_1<X_2<\dots<X_n\leq N(\epsilon)$ are all integers and $X$ has $\epsilon$-indistinguishable subsets and $N(\epsilon)=\exp_2^{n-2}(\Theta(\frac{1}{\epsilon}))$.
\end{coro}
\begin{proof}
Follows from Corollary \ref{coro:Theta} and Proposition \ref{prop:subsetfromnminus1}.
\end{proof}

  \section{Conclusion}
  We have shown that for any $n$ and $\epsilon>0$ there exists a distribution of $(X_1,\dots,X_n)$ with $1\leq X_1<\dots <X_n$ integers such that any two subsets of $\{X_1,\dots,X_n\}$ of the same size are $\epsilon$-indistinguishable. This could in theory be used to approximately time games the cannot be exactly timed. Unfortunately, the resulting values of $X_n$ are huge: Even for $n=4$ and $\epsilon=\frac{1}{200}$ we would need to use values much larger that the universe's age in Planck times.

 %[skj: to be done] $n=5, \epsilon=10^{-5}$: bigger foreheads than the observable universe  

  \bibliographystyle{plain} 
\bibliography{numbersonforehead}

\begin{thebibliography}{1}

\bibitem{Time}
Sune~K. Jakobsen, Troels Bjerre~S\o rensen, and Vincent Conitzer.
\newblock Timeability of extensive-form games.
\newblock 2015.
\newblock arxiv:1502.03430.

\bibitem{AGT07}
Noam Nisan, Tim Roughgarden, Eva Tardos, and Vijay~V. Vazirani.
\newblock {\em Algorithmic Game Theory}.
\newblock Cambridge University Press, New York, NY, USA, 2007.

\bibitem{PR97}
M.~Piccione and A.~Rubinstein.
\newblock On the interpretation of decision problems with imperfect recall.
\newblock {\em Games Econ. Beha}, 20:3--24, 1997.

\end{thebibliography}

	 \end{document}